\documentclass[12pt]{article}
\usepackage[utf8]{inputenc}
\usepackage{amsthm,amsmath,amssymb}
\usepackage[margin=1in]{geometry}
\usepackage{setspace}
\usepackage{graphicx}
\usepackage{subcaption,placeins}
\usepackage{xcolor}
\usepackage{newtxtext}
\usepackage{newtxmath}
\usepackage{xr}
\usepackage{hyperref}
\hypersetup{
    colorlinks=true,
    urlcolor=blue,
    citecolor=blue,
    linkcolor=blue
}
\PassOptionsToPackage{hyphens}{url}\usepackage{hyperref}
\PassOptionsToPackage{sloppy}{url}\usepackage{hyperref}
\usepackage{cleveref}
\usepackage{ctable}
\usepackage{indentfirst}
\usepackage{makecell}
\usepackage{booktabs}
\usepackage[normalem]{ulem}
\useunder{\uline}{\ul}{}
\usepackage{ragged2e}

\newcommand{\E}{\mathbb{E}}
\renewcommand{\P}{\mathrm{P}}

\usepackage[bibstyle=numeric, citestyle=authoryear, doi=false, url=true, backend=biber, maxbibnames=40, maxcitenames=4, uniquename=false, sorting=nyt]{biblatex}
\renewbibmacro{in:}{}
\DeclareNameAlias{default}{family-given/given-family}
\newcommand{\citet}{\textcite}
\renewcommand{\bibfont}{\small}

\newtheorem{assumption}{Assumption}
\newtheorem{proposition}{Proposition}
\newtheorem{lemma}{Lemma}

\theoremstyle{definition}
\newtheorem{remark}{Remark}

\crefname{assumption}{Assumption}{Assumptions}
\crefname{lemma}{Lemma}{Lemmas}

\title{Evaluating Policies Early in a Pandemic: Bounding Policy Effects with Nonrandomly Missing Data\thanks{Li dedicates this paper to the memory of his late mother, Mrs. Qing Zheng, who passed away on April 21, 2020.  This paper contains some results from our earlier paper  ``Understanding the Effects of Tennessee's Open Covid-19 Testing Policy: Bounding Policy Effects with Nonrandomly Missing Data'' and replaces that paper.} \thanks{We are grateful to the editor and three referees for their insightful comments that have greatly improved the paper.  We also thank Lesley Turner, John Weymark, Jun Zhao, Lily Zhao, and seminar participants at Vanderbilt University for helpful comments.}}

\author{Brantly Callaway\thanks{Department of Economics.  University of Georgia.  \href{mailto:brantly.callaway@uga.edu}{brantly.callaway@uga.edu}} \and Tong Li\thanks{Department of Economics.  Vanderbilt University.  \href{mailto:tong.li@vanderbilt.edu}{tong.li@vanderbilt.edu}}}

\newcommand\independent{\protect\mathpalette{\protect\independenT}{\perp}}
    \def\independenT#1#2{\mathrel{\setbox0\hbox{$#1#2$}%
    \copy0\kern-\wd0\mkern4mu\box0}} 

\begin{document}

\maketitle

\abstract{During the early part of the Covid-19 pandemic, national and local governments introduced a number of policies to combat the spread of Covid-19.  In this paper, we propose a new approach to bound the effects of such early-pandemic policies on Covid-19 cases and other outcomes while dealing with complications arising from (i) limited availability of Covid-19 tests, (ii) differential availability of Covid-19 tests across locations, and (iii) eligibility requirements for individuals to be tested.  We use our approach study the effects of Tennessee's expansion of Covid-19 testing early in the pandemic and find that the policy decreased Covid-19 cases.}

\vspace{150pt}


\bigskip


\vspace{70pt}

\pagebreak

\doublespacing

\normalsize

\addtolength{\jot}{-6pt}
\setlength{\abovedisplayskip}{\abovedisplayshortskip}
\setlength{\belowdisplayskip}{\belowdisplayshortskip}

\section{Introduction}

In response to the Covid-19 pandemic, governments implemented a large number of policies in order to try to slow the spread of Covid-19.  Most immediate policy responses were non-pharmaceutical interventions such as stay-at-home orders, school closures, mandates to wear masks, and contact tracing, among others.  In the current paper, we focus on trying to identify and estimate the effect of these sorts of policies.  We are interested in both the direct effect of these Covid-19 related policies on the number of Covid-19 cases as well as the effect of these policies on other outcomes (e.g., economic outcomes).  %

A main challenge with evaluating the effects of policies on the number of Covid-19 cases is that, especially during the early part of the pandemic when testing was not widely available, the confirmed number of cases in a particular location may have been quite different from the actual number of cases.\footnote{To be clear on the terminology, we use the phrase \textit{actual cases} or \textit{total cases} to refer to the total number of Covid-19 cases -- this is in general not observed.  We use the phrase \textit{confirmed cases} to refer to the number of positive tests for Covid-19.}  Besides limited testing early in the pandemic, another complication is that testing was non-random.   In particular (and especially early in the pandemic), individuals who were more likely to have had Covid-19 were more likely to have taken a Covid-19 test due to both self-selection as well as eligibility requirements for taking the test such as being in a particular risk group and/or exhibiting certain symptoms.  Moreover, the availability of tests varied across locations.  This creates the issue that, holding the number of cases fixed, ``more testing'' can mechanically lead to confirming more cases.  If the number of confirmed cases is treated as the actual number of cases, then these issues can lead to faulty inferences regarding the effects of policies on the actual spread of Covid-19.

In order to evaluate the effects of a Covid-19 related policy early in the pandemic, our strategy is to take treated locations and compare them to untreated locations with similar characteristics before the policy was implemented.  In particular, we compare locations that are similar along dimensions that are related to the spread of Covid-19 such as population size, per capita number of actual cases, per capita number of confirmed cases, and per capita tests that have been taken in that location.  This is an unconfoundedness-type strategy, %
and it would be a relatively straightforward exercise if testing were administered randomly; but, as discussed above, that was not the case.    In practice, this creates challenges along two dimensions.  First, ideally we would like to compare locations that had experienced a similar number of actual cases up to the time period when the policy was implemented.  In other words, for evaluating Covid-19 related policies, the actual number of cases in pre-treatment periods plays an important role as a covariate.  But conditioning on the number of pre-treatment cases is not directly possible because actual cases are not observed.  Second, understanding the effect of the policy on actual cases is of primary interest itself, which, of course, is more challenging when actual cases are not directly observed.

In the current paper, we propose a new approach for dealing with nonrandomly missing testing data.  We address the first challenge by introducing a conditional independence assumption that says that the distribution of pre-treatment actual Covid-19 cases is the same for treated and untreated locations conditional on having the same values of observed pre-treatment characteristics (the most important of which are confirmed cases, number of tests, and population size).  Under this sort of condition, the challenge of pre-treatment Covid-19 cases being an unobserved, but important, covariate is effectively solved; this further implies that, under this condition, the effect of a Covid-19 related policy on observed outcomes (such as economic outcomes) is point identified.

Addressing the second issue is more challenging.  Our intuition is that observing a policy that simultaneously decreases the number of confirmed cases while increasing (or at least not decreasing) the number of tests provides a strong indication that the policy is decreasing the actual number of Covid-19 cases in that location (even though these are not fully observed).  We formalize what extra conditions are required for this sort of intuition to hold.  In particular, the key extra condition is that the number of actual cases among the untested is not too large under the policy relative to what the number of actual cases among the untested would have been if the policy had not been implemented (note that the number of actual cases among the untested is not observed in either case); we also provide more primitive conditions for this extra condition to hold later in the paper.  Under these conditions, differences in (a scaled version of) the number of confirmed cases in a particular location relative to the number of confirmed cases among untreated locations with similar pre-treatment characteristics can be interpreted as an upper bound on the effect of the policy on the actual number of Covid-19 cases.  If this upper bound is negative, it suggests that the policy decreased the number of actual Covid-19 cases.

We apply our approach to studying the effects of Tennessee's expanded testing policy during the first few months of the pandemic.  One reason to study the effects of expanded Covid-19 testing early in the pandemic is that the costs of increased testing are substantially lower than other policy interventions such as stay-at-home orders or business closures.  However,  %
compared to other early pandemic policies, studying the effects of expanding access to Covid-19 tests is challenging due to the mechanical connections between testing and confirmed cases.  Early in the pandemic, Tennessee guaranteed payment for Covid-19 tests which resulted in a rapid expansion of testing capacity in the state.  Tennessee also implemented open-testing which allowed for any individual that wanted a test to be tested (this is in contrast to most states that, around this time, had eligibility requirements for being tested).  For observed outcomes, we find (i) the policy did notably increase the number of Covid-19 tests in Tennessee, (ii) it decreased the number of confirmed cases, and (iii) suggestive evidence that it led to a moderate increase in the number of trips to work (with the timing corresponding to the timing of the policy's reduction of confirmed cases).  For the (unobserved) per capita number of actual  Covid-19 cases, we obtain informative bounds on the effect of the policy.   We estimate that the policy reduced the number of actual Covid-19 cases in Tennessee by at least 8000 by May 9 relative to what they would have been in the absence of the policy.  Given an infection fatality rate of 1\%, this indicates that Tennessee's policy saved upwards of 80 lives during the early part of the pandemic.  %

\bigskip

\noindent \textbf{Related Work:} 

Our paper is related to a large and rapidly growing literature evaluating the effects of Covid-19 related policies.  Some examples include \citet{courtemanche-garuccio-le-pinkston-yelowitz-2020}, \citet{dave-friedson-matsuzawa-mcnichols-sabia-2020}, \citet{dave-friedson-matsuzawa-sabia-2021}, \citet{dave-friedson-matsuzawa-sabia-safford-2020}, \citet{gapen-millar-blerina-sriram-2020},  \citet{glaeser-jin-leyden-luca-2021}, \citet{goolsbee-syverson-2021}, \citet{gupta-montenovo-nguyen-rojas-schmutte-simon-weinburg-2020}, \citet{juranek-zoutman-2021}, \citet{kong-prinz-2020}, \citet{mitze-kosfeld-rode-walde-2020}, and \citet{ziedan-simon-wing-2020}, among others; \citet{haber-et-al-2022} provides a recent review of empirical strategies used to evaluate Covid-19 related policies. Work specifically on the effects of Covid-19 testing includes \citet{di-balzi-carreras-onder-2020} which compares mortality across four regions in Italy that had different testing rates, and  \citet{acemoglu-makhdoumi-malekian-ozdaglar-2020}, \citet{atkeson-droste-mina-stock-2020}, and \citet{brotherhood-kircher-santos-tertilt-2020} though these latter papers are much different than the current paper in that they follow a macroeconomic approach involving model calibration.  Other methodological work on evaluating Covid-19 related policies includes \citet{allcott-boxell-conway-ferguson-gentzkow-goldman-2020}, \citet{callaway-li-2022b}, \citet{chernozhukov-kasahara-schrimpf-2021}, \citet{aleman-busch-ludwig-santaeulalia-2020},  \citet{goodman-marcus-2020}, \citet{gauthier-2021}, and  \citet{weill-stigler-deschenes-springborn-2021}.  All of these papers are quite distinct from ours as they propose approaches that result in point identification and are not generally as concerned as we are about nonrandomly missing testing data.  More closely related to our paper is \citet{manski-molinari-2020} which proposes an approach to bound the actual number of Covid-19 cases in the presence of nonrandomly missing testing data.  Other work involving partial identification in the context of Covid-19 includes \citet{toulis-2021,stoye-2022} (for actual cases) and \citet{depalo-2021} (for excess mortality).  Our approach expands these sorts of arguments in ways that are suitable for policy evaluation and deals with several distinct complications that show up in this context.  %

On the econometrics side, our approach is related to the work on partially identifying treatment effect parameters in the presence of sample selection such as \citet{lee-2009,lechner-melly-2010}.  Like those papers, the key issues here are that the outcome (in our case, whether or not an individual has had Covid-19) is only observed for some self-selected individuals (in our case, tested individuals) and that the policy can affect both the outcome and the set of individuals for whom the outcome is observed.  Unlike those papers, our main interest is in a different parameter, %
and we invoke different assumptions to obtain bounds.  More generally, the results in this section are related to a large literature on partial identification of treatment effect parameters.  See \citet{tamer-2010,molinari-2020} for recent reviews of this literature.

\section{Methodology} \label{sec:methodology}

This section discusses our methodological approach.  We start by briefly discussing bounds on actual Covid-19 cases in a particular location.  This discussion largely follows \citet{manski-molinari-2020} and discusses (mild) assumptions that have been made in this context to bound per capita actual Covid-19 cases early in the pandemic.  This part also illustrates several of the key issues regarding non-randomly missing data.  Then, we discuss our approach to (point-) identifying Covid-19 related policy effects on observed outcomes in the case where we would like to condition on pre-treatment actual Covid-19 cases for identification.  Finally, we discuss our contribution to developing bounds on the effect of the policy on per capita actual Covid-19 cases.  In both cases, our approach is able to deal with limited/nonrandom Covid-19 testing.

\subsection{Bounding Rates of (Unobserved) Actual Covid-19 Cases} \label{sec:descriptive-bounds}

To start with, we introduce some notation.  Define $C_{ilt}$ to be a binary variable indicating whether or not individual $i$ in location $l$ had had Covid-19 by time period $t$.  Similarly, define $T_{ilt}$ and $R_{ilt}$ to be binary variables indicating whether or not individual $i$ in location $l$ had been tested ($T_{ilt}$) or tested positive ($R_{ilt}$) for Covid-19 by time period $t$.  Our first goal is descriptive:  to learn about $\P(C_{ilt}=1)$, the fraction of the population in location $l$ that has had Covid-19 by period $t$.\footnote{To be clear about the notation here, we are averaging over all individuals in a particular location $l$ at time period $t$.}

Identifying the fraction of individuals that have had Covid-19 is challenging because (i) not all individuals have been tested and some untested individuals may have had Covid-19 and (ii) testing has not been randomly assigned which implies the probability of having Covid-19 may be substantially different among tested and untested individuals.    The goal of this section is to develop informative bounds on the per capita number of actual Covid-19 cases in a particular location at a particular time under plausible identifying assumptions.  In particular, following \citet{manski-molinari-2020} notice that
\begin{align} \label{eqn:covid-bound-1}
  \P(C_{ilt}=1) = \P(C_{ilt}=1|T_{ilt}=1) \P(T_{ilt}=1) + \P(C_{ilt}=1|T_{ilt}=0) \P(T_{ilt}=0)
\end{align}
which follows immediately by the law of total probability. Next, we consider each of the terms in \Cref{eqn:covid-bound-1} individually.  First, $\P(T_{ilt}=1)$ and $\P(T_{ilt}=0)$ are the the (observed) fraction of the population in location $l$ at time period $t$ that have been tested/untested for Covid-19.
Next, $\P(C_{ilt}=1|T_{ilt}=1)$ is the fraction of the population in location $l$  at time period $t$ that has had Covid-19 conditional on being tested.  It can be re-written as
\begin{align*}
  \P(C_{ilt}=1|T_{ilt}=1) &= \P(C_{ilt}=1 | T_{ilt}=1, R_{ilt}=1)\P(R_{ilt}=1|T_{ilt}=1) \\
                         & \hspace{10pt} + \P(C_{ilt}=1 | T_{ilt}=1, R_{ilt}=0) \P(R_{ilt}=0|T_{ilt}=1) \\
                         &= \P(R_{ilt}=1|T_{ilt}=1) + \P(R_{ilt}=0|T_{ilt}=1,C_{ilt}=1) \P(C_{ilt}=1|T_{ilt}=1)
\end{align*}
where the first equality holds by the law of total probability and the second equality holds because (i) $R_{ilt} = 1 \implies T_{ilt}=1$  (i.e., in order to test positive, an individual has to be tested), (ii) we suppose that the false positive rate of the test is equal to 0 which implies that $\P(C_{ilt}=1|R_{it}=1)=1$,\footnote{The false positive rate is given by $\P(C_{ilt}=0|R_{ilt}=1)$, and there is evidence that the false positive rates for PCR tests are extremely low; see, for example \citet{sethuraman-jeremiah-ryo-2020}.} and (iii) repeated application of the definition of conditional probability for the second term.  Then, rearranging implies that
\begin{align} \label{eqn:CcondT}
  \P(C_{ilt}=1|T_{ilt}=1) = \frac{\P(R_{ilt}=1|T_{ilt}=1)}{1-FNR}
\end{align}
where $\P(R_{ilt}=1|T_{ilt}=1)$ is the (observed) fraction of tests that have come back positive in location $l$ at time period $t$ and where  $FNR := \P(R_{ilt}=0|T_{ilt}=1,C_{ilt}=1)$ is the false negative rate of the test.  This is a property of the test, and we set the false negative rate to be equal to 0.25.\footnote{\citet{manski-molinari-2020} put bounds on a closely related term called the Negative Predictive Value of the test; we could similarly put bounds on the false negative rate of the test.  We do not do this in the current paper in order to mainly focus on the bounds arising from non-random testing.  In the results presented below, in general, the bounds are not very sensitive to different reasonable values of the false negative rate of the test.}  %
Finally, $\P(C_{ilt}=1|T_{ilt}=0)$ is the (unobserved) fraction of the population that have had Covid-19 but have not been tested in location $l$ by time period $t$.  This term is the hardest to identify, and we discuss plausible assumptions that lead to bounds on this term below.\footnote{Much research studying effects of Covid-19 related policies uses confirmed cases as the outcome of interest which implicitly sets this term equal to 0; however, as discussed above, this term is unlikely to be equal to zero due to asymptomatic cases and limited testing.}

\Cref{eqn:CcondT} says that the probability of having Covid-19 conditional on being tested is increasing in the fraction of positive tests and the false negative rate of the test.  It also implies that every term in \Cref{eqn:covid-bound-1} is identified except $\P(C_{ilt}=1|T_{ilt}=0)$.  Without employing some additional assumption on this term, the bounds on the rate of actual cases are given by
\begin{align} \label{eqn:covid-bound-1b}
  \frac{\P(R_{ilt}=1)}{1-FNR} \leq \P(C_{ilt}=1) \leq \frac{\P(R_{ilt}=1)}{1-FNR} + \P(T_{ilt}=0)
\end{align}

In most cases, these sorts of bounds would be extremely wide.  For example, in our application about Tennessee's expanded testing policy, for the whole state of Tennessee on May 9, $\P(R_{ilt}=1)$ was about 0.2\% and $\P(T_{ilt}=0)$ was about 95.2\% (i.e., about 4.8\% of Tennessee's population had been tested and about 0.2\% had a positive test).  If the only restriction on $\P(C_{ilt}=1|T_{ilt}=0)$ is that it is bounded between 0 and 1, then this will lead to extremely wide bounds on Covid-19 cases (essentially uninformative).  Instead (and continuing to follow \citet{manski-molinari-2020}), we make the following assumption.

\begin{assumption}[Covid-19 Bound for Untested Individuals] \label{ass:covid-bound} \singlespacing
  \begin{align*}
    \P(C_{ilt}=1|T_{ilt}=0) \leq \P(C_{ilt}=1|T_{ilt}=1)
  \end{align*}
\end{assumption}

\Cref{ass:covid-bound} says that the fraction of individuals who have had Covid-19 (in a particular location) is lower among the group of individuals who have not been tested than among those who have been tested.  This is a mild assumption.  This assumption is likely to hold for two reasons.  First, early in the pandemic, tests were predominantly given to individuals expressing Covid-19 symptoms.  Second, even in states (or time periods) where testing was available to anyone who wanted to take a test, (i) individuals expressing symptoms were still among those most likely to take the test and (ii) it seems likely that there was some self-selection into taking the test among individuals who thought they may have Covid-19 even if they did not have the right combination of symptoms to otherwise warrant a test.  It is also helpful to think about the limiting cases of the assumption.  $\P(C_{ilt}=1|T_{ilt}=0)=0$ in the case when no untested individuals have had Covid-19.  $\P(C_{ilt}=1|T_{ilt}=0) = \P(C_{ilt}=1|T_{ilt}=1)$ if the probability of having had Covid-19 is the same for individuals who have not been tested as for individuals who have been tested.  This condition would hold if testing were randomly assigned.  In practice, either of these limiting conditions would be strong enough to point identify $\P(C_{ilt}=1)$; however, based on the above discussion, neither of these limiting conditions seems likely to hold.  Instead, \Cref{ass:covid-bound} imposes the much weaker condition that the probability of having had Covid-19 for the group of individuals who have not been tested falls in between these two limiting cases.

\Cref{ass:covid-bound} does not affect the lower bound on the number of actual cases, but it is potentially very useful in lowering the upper bound on the  number of Covid-19 cases in a particular location.  In particular, notice that under \Cref{ass:covid-bound},
\begin{align} \label{eqn:covid-bound-2}
  \P(C_{ilt}=1) \leq \P(C_{ilt}=1|T_{ilt}=1) 
\end{align}
This can lead to a much tighter bound especially when $\P(C_{ilt}=1|T_{ilt}=1)$ is substantially less than one.  For example, in our application, for the whole state of Tennessee, $\P(C_{ilt}=1|T_{ilt}=1)$ is roughly equal to 6\% on May 9.  This immediately leads to a much tighter bound on the number of actual cases relative to not putting any restrictions on $\P(C_{ilt}=1|T_{ilt}=0)$.  

\subsection{Policy Evaluation with Nonrandomly Missing Data}

The previous section discussed how to bound the number of actual Covid-19 cases in a particular location.  The main goal of the paper is to go beyond these descriptive bounds and evaluate how a policy affects the (unobserved) number of actual Covid-19 cases as well as other outcomes such as confirmed cases and trips to work.  We discuss our approach in this section.

For this section, our arguments are about policy effects in certain locations and, therefore, we slightly modify the notation from the previous section.  In particular, define $C_{lt} := \P(C_{ilt}=1)$, $R_{lt} := \P(R_{ilt}=1)$, $T_{lt} := \P(T_{ilt}=1)$.  These are defined for a particular location (rather than for a particular individual) and correspond to the fraction of the population in location $l$ that has had Covid-19, that have tested positive for Covid-19 (i.e., the per capita number of confirmed cases), and that have been tested for Covid-19, respectively.\footnote{Also, notice that we do not need to estimate $R_{lt}$ and $T_{lt}$; rather each of them is exactly observed.}  We also suppose that we have access to location-level covariates $X_l$ that do not vary over time; in practice, the most important covariate is the total population in a particular location.  Some of the results below consider policy effects on other outcomes; in that case we denote the location-level outcome in time period $t$ by $Y_{lt}$ (e.g., the number of deaths or the number of trips to work).

In order to think about policy effects, we define potential outcomes for location $l$ in time period $t$.  In particular, let $C_{lt}(1)$, $T_{lt}(1)$, $R_{lt}(1)$, and $Y_{lt}(1)$ denote the per capita number of actual Covid-19 cases, the per capita number of tests, the per capita number of confirmed cases, as well as other outcomes that would occur in location $l$ in time period $t$ if the policy were in place.  Similarly, if the policy is not in place for location $l$ in time period $t$, we denote the untreated potential outcomes that would occur in this case by: $C_{lt}(0)$, $T_{lt}(0)$, $R_{lt}(0)$, and $Y_{lt}(0)$.  To conserve on notation, define $Z_{lt}(d) = (Y_{lt}(d), R_{lt}(d), T_{lt}(d), X_l')'$ for $d \in \{0,1\}$.  This collects the covariates and all potential outcomes except for $C_{lt}(d)$. Also, define $Z^*_{lt}(d) = (Z_{lt}(d)',C_{lt}(d))'$ which additionally includes per capita actual cases.

Next, let $D_l$ be a binary variable indicating treatment participation.  For locations that participate in the policy, $D_l=1$; otherwise, $D_l=0$.  Also suppose that  there are two time periods: $t^*$ and $t^*-1$,\footnote{Our results extend immediately to the case where there are more available time periods.} and that the policy is implemented between time periods $t^*$ and $t^*-1$.  In this setup, we observe
\begin{align*}
  Z_{lt^*} = D_l Z_{lt^*}(1) + (1-D_l)Z_{lt^*}(0) \quad \textrm{and} \quad Z_{lt^*-1} = Z_{lt^*-1}(0)
\end{align*}
In other words, in post-treatment time periods we observe treated potential outcomes for locations that participate in the treatment and observe untreated potential outcomes for locations that do not participate in the treatment.  In pre-treatment time periods, we observe untreated potential outcomes for all locations. %

\subsubsection{Policy Effects on Observed Outcomes}
To start with, consider identifying the effect of a Covid-19 related policy on some observed outcome (e.g., the number of Covid-19 tests, confirmed cases, or  trips to work) in location $l$ at time period $t^*$.  We start with this case because it is simpler as $Y_{lt}$, the outcome, is fully observed while $C_{lt}$, the per capita number of actual cases in location $l$, is not.  Our interest in this section is in identifying
\begingroup
\small
\begin{align*}
  ATT_Y(Z_{lt^*-1}) = \E[ Y_{lt^*}(1) - Y_{lt^*}(0) | Z_{lt^*-1}, D_{l}=1] \quad \textrm{and} \quad ATT_Y = \E[Y_{lt^*}(1) - Y_{lt^*}(0)|D_l=1]
\end{align*}
\endgroup
$ATT_Y(Z_{lt^*-1})$ is the average effect of the  policy on the outcome among treated locations with pre-treatment characteristics $Z_{lt^*-1}$.  $ATT_Y$ is the overall average effect of the policy among treated locations.  We make the following assumption
\begin{assumption}[Unconfoundedness] \label{ass:unc} \singlespacing
  \begin{align*}
    \E[Y_{lt^*}(0) | Z^*_{lt^*-1}(0), D_l=1] = \E[Y_{lt^*}(0) | Z^*_{lt^*-1}(0), D_l=0]
  \end{align*}
\end{assumption}

\Cref{ass:unc} is a standard and widely used assumption to identify the effect of some policy (see, for example, \citet{imbens-wooldridge-2009}).  It says that, if the policy had not been enacted, on average, outcomes in treated locations would have been the same as outcomes in untreated locations that had the same pre-treatment characteristics; i.e., the same outcomes in the previous period, the same per capita number of confirmed cases, the same number of per capita tests, the same population, \textit{as well as the same per capita number of actual cases}.\footnote{It is also worth pointing out that \Cref{ass:unc} is not a main requirement of our approach.  Depending on the particular outcome of interest, one could employ an alternative baseline identification strategy.  To give one particular leading example, one could replace the unconfoundedness assumption in \Cref{ass:unc} with a conditional parallel trends assumptions (i.e., that $\E[\Delta Y_{lt^*}(0) | Z^*_{lt^*-1}(0), D_l=1] = \E[\Delta Y_{lt^*}(0) | Z^*_{lt^*-1}(0), D_l=0]$) and then show an analogous result to \Cref{prop:1} but with $\Delta Y_{lt^*}$ replacing $Y_{lt^*}$ everywhere in that proposition (note that, in this case, $Z_{lt^*-1}(0)$ ought to be modified to not include the lagged outcome).  Thus, the relatively more important assumption in this section is the one in \Cref{ass:mar}.}

One cannot immediately use \Cref{ass:unc} because $Z^*_{lt^*-1}(0)$ includes $C_{lt^*-1}(0)$ --- the per capita number of actual Covid-19 cases in a particular location --- which is unobserved.  But, in practice, many outcomes in period $t^*$ are likely to depend on how widespread Covid-19 has been -- even if it has gone largely undetected.  Therefore, it seems quite important to control for the (unobserved) number of cases.  To address this issue, we make the following assumption

\begin{assumption}[Conditional Independence of Pre-Policy Actual Covid-19 Cases] \label{ass:mar} \singlespacing
  \begin{align*}
    C_{lt^*-1}(0) \independent D_l | Z_{lt^*-1}(0)
  \end{align*}
\end{assumption}
\Cref{ass:mar} says that, in the pre-treatment period, the distribution of the per capita number of actual Covid-19 cases was the same among treated and untreated locations that had the same pre-policy characteristics.  %
It rules out systematic differences in unobserved total cases in the pre-treatment period among treated and untreated locations with similar populations and that had run a similar number of tests and had confirmed a similar number of cases.%

\begin{proposition} \singlespacing \label{prop:1}  Under \Cref{ass:unc,ass:mar}, $ATT_Y(Z_{lt^*-1})$ and $ATT_Y$ are identified and given by
  \begin{align*}
    ATT_Y(Z_{lt^*-1}) = \E[Y_{lt^*} | Z_{lt^*-1}, D_l=1] - \E[Y_{lt^*} | Z_{lt^*-1}, D_l=0]
  \end{align*}
  and
  \begin{align*}
    ATT_Y = \E[ATT_Y(Z_{lt^*-1})|D_l=1]
  \end{align*}
\end{proposition}

The proof of \Cref{prop:1} is provided in \Cref{sec:proofs}.  The result in \Cref{prop:1} says that the average effect of the policy among treated locations is point identified even in the case where the outcomes themselves could depend on the number of actual cases and the number of actual cases is not observed.  The main conditions for this identification result are (i) some baseline identification strategy (we used unconfoundedness) and (ii) the conditional independence assumption for pre-treatment actual Covid-19 cases.  %
Moreover, $ATT_Y$ can be recovered by comparing outcomes in treated locations to outcomes in untreated locations that had the same observed pre-treatment characteristics.

\subsubsection{Policy Effects on (Unobserved) Actual Covid-19 Cases}

Next, we consider trying to identify the effect of a policy on the per capita number of actual Covid-19 cases.
This is distinctly more challenging than the previous case because the number of actual cases is not observed.  Relative to the descriptive bounds presented in \Cref{sec:descriptive-bounds}, developing bounds on policy effects introduces new challenges.  For example, suppose that there are two locations, one treated and one untreated, and we are (i) comfortable with the idea that number of cases experienced in the untreated location is equal to the number of cases that the treated location would have experienced if it had not implemented the policy but (ii) are only able to bound the actual number of cases in each location.  In this case, without further assumptions, the bounds on the effect of the policy are equal to the difference between the upper (or lower) bound for the treated location minus the lower (or upper) bound for the untreated location.  These sorts of bounds are generally very wide and likely to cover 0.  We introduce some additional assumptions that are able to substantially narrow these sorts of bounds; the assumptions that we introduce formalize the idea that if both confirmed cases decrease in the treated location relative to the untreated location while testing does not decrease in the treated location relative to the untreated location, then this would be strong evidence that the policy led to a decrease in actual Covid-19 cases relative to what they would have been had the policy not been implemented.  We also formalize what other conditions need to hold in order for this intuition to be correct.

To start with, we continue to make \Cref{ass:mar}, and we modify \Cref{ass:unc} to hold jointly for all untreated potential outcomes and covariates:
\begin{assumption}  [Covid Unconfoundedness] \singlespacing \label{ass:covid-unc}
  \begin{align*}
    Z^*_{lt^*}(0) \independent D_l | Z^*_{lt^*-1}(0)
  \end{align*}
\end{assumption}
\Cref{ass:covid-unc} is similar to \Cref{ass:unc} in that it is an unconfoundedness type of assumption, but it applies to all untreated potential outcomes and covariates.  It says that, in the absence of the policy intervention,  the distribution of pandemic related variables (e.g., per capita actual cases, confirmed cases, and tests) would have been the same for treated locations and untreated locations conditional on having the same pre-treatment characteristics.  This type of assumption is compatible with SIR (which stands for Susceptible, Infected, Recovered) epidemic models which are the most prominent type of models for studying pandemics (see, for example, \citet{kermack-mckendrick-1927,allen-2008,allen-2017} in general and \citet{oka-wei-zhu-2021,fernandez-jones-2022,ellison-2020,acemoglu-chernozhukov-werning-whinston-2021,bisin-moro-2022} in economics).  \citet{chernozhukov-kasahara-schrimpf-2021,allcott-boxell-conway-ferguson-gentzkow-goldman-2020,callaway-li-2022b} provide connections between between epidemic models and various policy evaluation strategies.  For example, the relative merits of unconfoundedness-type identifying assumptions compared to difference in differences-type identifying assumptions are discussed at length in \citet{callaway-li-2022b} with that paper generally arguing in favor of unconfoundedness rather than difference in differences due to the high degree of nonlinearity (and lack of additively separable location-specific unobserved heterogeneity) in epidemic models.  %

Similarly to the previous section, we focus on identifying
\begingroup
\small
\begin{align*}
  ATT_C(Z_{lt^*-1}) = \E[C_{lt^*}(1) - C_{lt^*}(0) | Z_{lt^*-1}, D_l=1] \quad \textrm{and} \quad ATT_C = \E[C_{lt^*}(1) - C_{lt^*}(0) | D_l=1]
\end{align*}
\endgroup
$ATT_C(Z_{lt^*-1})$ is the average effect of the policy on the per capita number of actual Covid-19 cases across treated locations with pre-treatment characteristics $Z_{lt^*-1}$.  $ATT_C$ is the overall average effect of the policy on the per capita number of actual Covid-19 cases across treated locations.  In addition, the same sorts of arguments as in the previous section continue to go through suggesting that
\begin{align} \label{eqn:attc}
  ATT_C(Z_{lt^*-1}) &= \E[C_{lt^*}|Z_{lt^*-1},D_l=1] - \E[C_{lt^*}|Z_{lt^*-1},D_l=0] \nonumber \\
                    &= \P(C_{ilt^*} = 1| Z_{lt^*-1}, D_l=1) - \P(C_{ilt^*} = 1 | Z_{lt^*-1}, D_l=0) 
\end{align}
and
\begin{align*}
  ATT_C &= \E[C_{lt^*}|D_l=1] - \E\Big[ \E[C_{lt^*}|Z_{lt^*-1},D_l=0] | D_l=1 \Big] 
\end{align*}
The problem here is that $C_{lt^*}$ is not directly observed, and, therefore, as in \Cref{sec:descriptive-bounds}, the terms in \Cref{eqn:attc} are only partially identified.  Thus, our approach is to construct bounds on $ATT_C$.  

Before stating these results, we define three more terms to conserve on notation below.  First, for $d \in \{0,1\}$, define
\begin{align*}
      \gamma_d(Z_{lt^*-1}) := \frac{\P(R_{ilt}=1|Z_{lt^*-1},D=d)}{1-FNR}
\end{align*}
$\gamma_d(Z_{lt^*-1})$ is a scaled version of the number of confirmed cases in location $l$ in time period $t$.  This term is point identified as we observe the number of confirmed cases and know the false negative rate of the test.  Notice that this term corresponds to the first term in \Cref{eqn:covid-bound-1} (now conditional on $Z_{lt^*-1}$ and $D_l=d$).\footnote{For some of the expressions below and in the proofs, it is also helpful to notice that $\gamma_d(Z_{lt^*-1}) = \P(C_{ilt^*}=1, T_{ilt^*}=1 | Z_{lt^*-1}, D_l=d)$ (which follows from the same sorts of arguments as in \Cref{sec:descriptive-bounds}).}  Second, define
\begin{align*}
  \tau_{d}(Z_{lt^*-1}) := \P(T_{ilt^*}=1|Z_{lt^*-1}, D_l=d)
\end{align*}
which is the probability of being tested conditional on a location's pre-treatment characteristics and treatment status.  This is identified by the sampling process since we observe the number of tests in a particular location.  Finally, define
\begin{align*}
    \phi_d(Z_{lt^*-1}) := \P(C_{ilt^*}=1|T_{ilt^*}=0,Z_{lt^*-1}, D_l=d)
\end{align*}
which is the probability of having had Covid-19 conditional on having not been tested, a location's pre-treatment characteristics, and treatment status.  As above, this term is not identified because we do not observe the rate of Covid-19 cases among untested individuals.  Next, we provide an intermediate result that decomposes $ATT_C(Z_{lt^*-1})$ that is useful for developing our main bounds later in this section.

\begin{proposition} \singlespacing \label{prop:attc-decomp} Under \Cref{ass:mar,ass:covid-unc},
\begingroup
\small
  \begin{align*}
      ATT_C(Z_{lt^*-1}) = \big(\gamma_1(Z_{lt^*-1}) - \gamma_0(Z_{lt^*-1}) \big) + (\phi_1(Z_{lt^*-1})(1-\tau_1(Z_{lt^*-1})) - \phi_0(Z_{lt^*-1})(1-\tau_0(Z_{lt^*-1})) \big)
  \end{align*}
  \endgroup
\end{proposition}
The proof of \Cref{prop:attc-decomp} is provided in \Cref{sec:proofs}, but it is worth pointing out that it holds almost immediately by (i) the unconfoundedness assumption and the conditional independence assumption for pre-treatment actual cases combined with (ii) similar calculations to the ones in \Cref{sec:descriptive-bounds}.  It is also worth noting that all the terms in \Cref{prop:attc-decomp} are identified except $\phi_1$ and $\phi_0$ which are the probability of having had Covid-19 conditional on having not been tested (and pre-treatment location characteristics) for the treated group and untreated group, respectively.

Next, we discuss bounds on $ATT_C$ under the additional condition in \Cref{ass:covid-bound} (recall that this assumption says that the probability of having had Covid-19 among those who have not been tested is less than or equal to the probability of having had Covid-19 conditional on being tested).  

\begin{proposition} \singlespacing \label{prop:2} Under \Cref{ass:covid-bound,ass:mar,ass:covid-unc},
  \begin{align*}
    C^{B,L}_{lt^*}(Z_{lt^*-1}) \leq ATT_C(Z_{lt^*-1}) \leq C^{B,U}_{lt^*}(Z_{lt^*-1})
  \end{align*}
  where
  \begin{align*}
    C_{lt^*}^{B,L}(Z_{lt^*-1}) &:= \gamma_1(Z_{lt^*-1}) - \gamma_0(Z_{lt^*-1}) - \gamma_0(Z_{lt^*-1}) \frac{1-\tau_{0}(Z_{lt^*-1})}{\tau_{0}(Z_{lt^*-1})} \\[10pt] %
    C_{lt^*}^{B,U}(Z_{lt^*-1}) &:= \gamma_1(Z_{lt^*-1}) - \gamma_0(Z_{lt^*-1}) + \gamma_1(Z_{lt^*-1}) \frac{1-\tau_{1}(Z_{lt^*-1})}{\tau_{1}(Z_{lt^*-1})}%
  \end{align*}
  and
  \begin{align*}
    \E\Big[ C^{B,L}_{lt^*}(Z_{lt^*-1}) | D_l=1 \Big] \leq ATT_C \leq \E\Big[ C^{B,U}_{lt^*}(Z_{lt^*-1}) | D_l=1\Big]
  \end{align*}
\end{proposition}

The proof of \Cref{prop:2} is provided in \Cref{sec:proofs}.  %
The term in common for each of the bounds, $\gamma_1(Z_{lt^*-1}) - \gamma_0(Z_{lt^*-1})$, is driven by differences in confirmed cases among treated and untreated locations with similar pre-treatment characteristics.  The extra term for the lower bound comes from setting the fraction of untested individuals in treated locations who have had Covid-19 to be equal to zero while setting the fraction of untested individuals in untreated locations who have had Covid-19 to be equal to the fraction who have had Covid-19 conditional on being testing (this comes from the bound in \Cref{ass:covid-bound}).  The upper bound comes from doing the opposite.%
\footnote{In practice, the extreme cases that lead to the lower bound and upper bound seem unlikely to hold.  This suggests that these bounds are likely to be quite conservative.}  %
The weights on these terms (the terms involving $\tau_1$ and $\tau_0$) also tend to be very large because the fraction of untested individuals is much larger than the fraction of tested individuals which implies that $\tau_d(Z_{lt^*-1}) \ll 1 - \tau_d(Z_{lt^*-1})$ for $d \in \{0,1\}$.

The drawback of these bounds is that they are unlikely to be informative about the sign of the policy effect.  To see this, notice that, especially early in the pandemic, the terms involving $\gamma_d(Z_{lt^*-1})$ are often quite small due to the number of confirmed cases being small.  On the other hand, the extra terms can be orders of magnitude larger.  %
In our application, using these bounds, the bounds cover 0 in all time periods and are not very informative.

In order to proceed, it is helpful to re-write the expression for $ATT_C(Z_{lt^*-1})$ in \Cref{prop:attc-decomp} as
\begin{align} 
    ATT_C(Z_{lt^*-1}) &= \big(\gamma_1(Z_{lt^*-1}) - \gamma_0(Z_{lt^*-1}) \big) \label{eqn:attc_decomp2-1}\\ 
    & + \big(\phi_1(Z_{lt^*-1}) - \phi_0(Z_{lt^*-1}) \big)(1-\tau_1(Z_{lt^*-1})) \label{eqn:attc_decomp2-2}\\
    & + \phi_0(Z_{lt^*-1})\big(\tau_0(Z_{lt^*-1}) -\tau_1(Z_{lt^*-1}) \big) \label{eqn:attc_decomp2-3}
\end{align}
Recall that the only terms in this expression that are not identified are $\phi_1$ and $\phi_0$.  In the case where a researcher is interested in trying to determine whether or not the policy decreased actual Covid-19 cases, one would be interested in determining if $ATT_C$ is less than or equal to 0.  The term in \Cref{eqn:attc_decomp2-1} comes from the difference between the number of confirmed cases in treated locations relative to untreated locations with similar pre-treatment characteristics and is point identified.  %
Neither of the expressions in Equation (\ref{eqn:attc_decomp2-2}) or (\ref{eqn:attc_decomp2-3}) is point identified.  However, the sign of the term in \Cref{eqn:attc_decomp2-3} is fully determined by the difference between the number of tests under the policy relative to the number tests among similar untreated locations.  If the treatment does not decrease the number of tests, then this term will be non-positive.  Finally, the sign of the term in \Cref{eqn:attc_decomp2-2} depends on the difference between the number of cases among untested individuals in treated locations relative to the number of cases among untested individuals in similar untreated locations.  This difference is not identified and is the most challenging part to think through (we address this point in substantially more detail below).  For now though, it is important to note that the above expression clarifies our intuition from the introduction: a policy that decreases the number of confirmed cases while not decreasing the number of tests does in fact decrease the number of actual cases as long as the number of cases among untested individuals under the policy does not increase relative to what the number of cases among untested individuals would have been if the policy had not been implemented.

To formalize the above discussion, we introduce the following assumption.
\begin{assumption} [Bound on Actual Cases and Untested Individuals] \label{ass:pol-bound} \singlespacing
\begingroup
\small
  \begin{align*}
    \P(C_{ilt^*}(1) = 1, T_{ilt^*}(1) = 0, | Z_{lt^*-1}(0), D_l=1) \leq \P(C_{ilt^*}(0) = 1, T_{ilt^*}(0) = 0 |  Z_{lt^*-1}(0), D_l=1)
  \end{align*}
\endgroup
\end{assumption}

\Cref{ass:pol-bound} says that, for individuals in locations that experience the policy, the joint probability of having had  Covid-19 and not being tested under the policy is less than or equal to the joint probability of having Covid-19 and not being tested in the absence of the policy.  Under \Cref{ass:mar,ass:covid-unc}, this is equivalent to saying that the sum of the terms in \Cref{eqn:attc_decomp2-2,eqn:attc_decomp2-3} is less than or equal to 0.  We provide one set of more primitive conditions and more detailed discussion of this assumption in the Supplementary Appendix.  There, we show that this assumption holds under the conditions that (i) the policy does not make tests less available than they otherwise would have been without the policy, (ii) the policy does not increase the number of Covid-19 cases among the fixed group of individuals who would not be tested under the policy relative to what they would have been without the policy, (iii) there is not negative selection into taking the test for ``compliers'' (those who would be tested under the policy but would not be tested in the absence of the policy) relative to ``never-takers'' (those who are not tested with or without the policy); that is, the probability of having Covid-19 is at least as high among ``compliers'' as among ``never-takers.''

Finally for this section, we provide our main result on tighter bounds for evaluating policies early in a pandemic.

\begin{proposition} \singlespacing \label{prop:3} Under \Cref{ass:covid-bound,ass:mar,ass:covid-unc,ass:pol-bound},
  \begin{align*}
    C^{C,L}_{lt^*}(Z_{lt^*-1}) \leq ATT_C(Z_{lt^*-1}) \leq C^{C,U}_{lt^*}(Z_{lt^*-1})
  \end{align*}
  where
  \begin{align*}
    C_{lt^*}^{C,L}(Z_{lt^*-1}) &:= C_{lt^*}^{B,L}(Z_{lt^*-1}) \\[10pt]
    C_{lt^*}^{C,U}(Z_{lt^*-1}) &:= \gamma_1(Z_{lt^*-1}) - \gamma_0(Z_{lt^*-1})
  \end{align*}  
  and
  \begin{align*}
    \E\Big[ C^{C,L}_{lt^*}(Z_{lt^*-1}) | D_l=1 \Big] \leq ATT_C \leq \E\Big[ C^{C,U}_{lt^*}(Z_{lt^*-1}) | D_l=1\Big]
  \end{align*}
\end{proposition}

The proof of \Cref{prop:3} is provided in \Cref{sec:proofs}.  Notice that the lower bound is the same as it was in the previous case, but that the upper bound can be substantially tighter.  In particular, the upper bound does not contain the same extra term as in \Cref{prop:2}; as discussed earlier, this term is the ``dominant'' term in the upper bound, and it is removed under the additional condition in \Cref{ass:pol-bound}.  This result provides conditions under which a policy that leads to a decrease in confirmed cases also indicates that the policy decreased actual cases. 

\begin{remark} \label{rem:shorter-bounds} \onehalfspacing
    The assumptions discussed in this section are useful for tightening the upper bound on policy effects which, therefore, targets answering the question of whether or not the policy reduced the number of actual Covid-19 cases.  In the Supplementary Appendix, we consider an additional type of assumption, which says that the probability of having Covid-19 among the untested is not ``too different'' with or without the policy, that can be useful for tightening the lower bound.
\end{remark}

\begin{remark} \onehalfspacing
    Even if the researcher is willing to assume that the probability of having had Covid-19 is the same, but not exactly known, among the untested for both the treated group and the untreated group (which is likely to be a very strong assumption), in general, this would still lead to bounds on policy effects.  We provide more details for this point in the Supplementary Appendix.%
\end{remark}

\section{Application: Tennessee's Open-Testing Policy}

 \label{sec:policy}

In this section, we study the effects of Tennessee's Covid-19 testing policies early in the pandemic.  Over the first few months of the pandemic, Tennessee was one of the most aggressive states in terms of making Covid-19 tests widely available for state residents.  However, studying the effects of testing availability itself is challenging for the reasons discussed above: (i) holding fixed the number of actual cases, more testing can mechanically lead to confirming more cases, and (ii) non-random selection into being tested.  Relative to other states, one important policy that led to the widespread availability of tests was that Tennessee directly paid private labs for processing tests resulting in private labs in the state quickly ramping up testing capacity (\citet{farmer-2020}).  Although Tennessee already had a high-level of testing, its distinctive early policy was its open-testing policy that  simultaneously increased the availability of tests and relaxed all eligibility requirements for obtaining a test.    Tennessee announced its open-testing policy on April 15, 2020 (\citet{tennessean-2020}), and over the weekends of April 18, April 25, and May 2, more than 23,000 individuals were tested at a total of 67 different testing sites (\citet{tngov-2020}).  Following those three weekends, Tennessee modified its open-testing policy to emphasize testing high-risk populations; that being said, relative to most other states, the requirements to be tested for Covid-19 in Tennessee continued to be low even after the policy was modified (\citet{hartnett-2020}).  We focus on the effects of Tennessee's testing policy through May 9 which is a week after the open-testing policy ended and as Tennessee moved to an alternative testing strategy targeting high risk groups.  We provide some additional details about Tennessee's policy in the Supplementary Appendix.

Importantly, for Tennessee, we are able to address a number of additional challenges for studying policies implemented early in the pandemic.  First, we document below that, besides having more widespread availability of tests, the other policies implemented in Tennessee (e.g., stay-at-home orders and school closures, among others) were quite similar to the policies of its surrounding states.  Second, we use disaggregated county-level data rather than more aggregated state-level data.  An important aspect of our identification strategy is to compare locations that had experienced similar pandemics locally prior to the policy being implemented (as well as having similar other characteristics such as population).  In general, we are able to make substantially more suitable comparisons at the county-level than at the state level.%

Our main results below use county-level data, but to start with, we provide some descriptive, state-level data which is presented in \Cref{fig:state-tests-cases}.  Relative to Alabama, Arkansas, Georgia, Kentucky, Mississippi, and North Carolina (these are the states that make up our comparison group below and are all states that border Tennessee),\footnote{Tennessee also shares small borders with Missouri and Virginia but these states tend to be geographically further away, and we do not use them as part of the comparison group below.} by the end of March, Tennessee was conducting more per capita tests per day than any of these states.  It is useful to point out that this was quite early in the pandemic.  Tennessee's first confirmed Covid-19 case was on March 5 (\citet{tngov-2022}), and by the end of the month Tennessee had 1,981 confirmed cases.  That being said, the big divergence between testing in Tennessee and its surrounding states corresponded to the beginning of the open-testing starting the weekend of April 18.  For example, on April 25, Tennessee was running between 50\% and 80\% more tests per capita than its surrounding states (see panel (b) of \Cref{fig:state-tests-cases}).  On the other hand, Tennessee was closer to the middle in terms of number of confirmed cases over the entire period that we consider (see panels (c) and (d) of \Cref{fig:state-tests-cases}).\footnote{The noticeable spike in confirmed cases in Tennessee around May 2 is driven by large Covid-19 outbreaks in prisons in Trousdale and Bledsoe counties (\citet{allison-timms-2020} and \citet{timms-2020}).  In some of our descriptive analysis, we keep these counties, but in our main results, we drop these counties.}  We provide an analogous figure to \Cref{fig:state-tests-cases} but for counties in Tennessee in \Cref{fig:county-tests-cases} in the Supplementary Appendix.\footnote{\Cref{fig:state-tests-cases} also illustrates the usefulness of using county-level data relative to state-level data.  Even among its immediate surrounding states, there are notable differences in terms of testing and confirmed cases already in March 2020; in terms of testing and confirmed cases, the most similar state is Alabama, but Tennessee has a 70\% higher population density than Alabama suggesting limited usefulness of using Alabama as a comparison state at the aggregate, state-level.  Taken together, this suggests that, at the state-level, the states bordering Tennessee are not similar enough to reliably use to produce counterfactual pandemics for Tennessee.  That said, it is much less demanding to take a particular county in Tennessee and to find counties in surrounding states that were experiencing similar pandemics in pre-treatment periods.}

\subsection{Data} \label{sec:data}

The main data that we use consists of county-level data on tests, confirmed cases, and deaths for Tennessee and its surrounding states from the Centers for Disease Control and Prevention.  The CDC collects local data from state and/or local health departments and provides it in a unified format.  For our purposes, the most challenging variable to collect is county-level testing which, to our knowledge, is not directly available over time at the county-level; we scraped this data county by county from the CDC's Covid Data Tracker Integrated County View.\footnote{See  \url{https://covid.cdc.gov/covid-data-tracker/\#county-view?list_select_state=Georgia&data-type=CommunityLevels&list_select_county=13219&null=CommunityLevels} for an example of the county-level reports available from the CDC.}  The testing data is for Nucleic Acid Amplification Tests (NAATs); this includes PCR tests through laboratories (including public, commercial, and hospital laboratories, among others) but does not include antibody or antigen tests.  A detailed discussion of how the CDC collects data from laboratories is available at \url{https://www.cdc.gov/coronavirus/2019-ncov/lab/reporting-lab-data.html}.  Similarly, confirmed cases and deaths originate from state and/or local health departments which we also collected from the CDC's county-level reports.  The date of confirmed cases corresponds to the day in which it was reported as a confirmed case.  The CDC provides the seven-day moving average for tests, confirmed cases, and deaths.  For many of the results below, we scale this by the county population and multiply by 1000 so that each county-level variable is per 1000 people in the county.  We also provide a number of results in terms of cumulative tests and confirmed cases which are derived from the seven-day moving average variables.  

The CDC suppresses some county-level values of confirmed cases and deaths in order to protect patient privacy.\footnote{The county-level CDC data that we use is derived from several underlying datasets.  For example, the data on confirmed cases builds on the CDC's Community Transmission Data (\url{https://data.cdc.gov/Public-Health-Surveillance/United-States-COVID-19-County-Level-of-Community-T/8396-v7yb}).  County-level confirmed cases are suppressed if, in the previous seven days, there have been more than zero new confirmed cases,  but less than 0.1 new cases per 1000 people in a county.}  When particular data is suppressed, it is for small, positive values of confirmed cases and/or deaths (i.e., we observe when the number of confirmed cases or deaths is equal to 0); and, for example, in our data, the smallest, non-zero observed seven day moving average of confirmed cases in a county is equal to 1.4.  For observations where the number of confirmed cases is suppressed, we set it equal to 1; and for observations where the number of deaths is suppressed, we set it equal to 0.5.  These choices result in the county-level data ``adding up'' to produce similar numbers of confirmed cases and deaths as are available at the state-level.  For the set of states and time periods that we consider, 19\% of county-level confirmed cases are suppressed and 43\% of deaths are suppressed; we observe full information about testing.  Finally, we merge the county-level testing and confirmed cases data with data from the Census Bureau on county-level population. 

\subsection{Estimation}

The identification results discussed in \Cref{sec:methodology} are constructive and suggest plug-in estimators of each parameter of interest.  In principle, a number of estimation procedures (e.g., regression, matching, or inverse probability weighting, among others) would be suitable for our proposed approach.  This section describes the particular doubly robust estimation procedure that we use in the application that comes from \citet{kang-schafer-2007}.\footnote{This approach is closely related to, but slightly different from, the more common augmented inverse propensity score weighting approach which is also doubly robust; for example, \citet{robins-rotnitzky-zhao-1994,scharfstein-rotnitzky-robins-1999,sloczynski-wooldridge-2018,callaway-li-2022b}.}  %
The first step of the estimation procedure is to estimate a propensity score model $p(z) := \P(D_l=1|Z_{lt^*-1}=z)$; for this step, we use logit and include county-level population and the seven day lags of cumulative tests, confirmed cases, and deaths per 1000 people in the county as covariates.  Given this estimate of the propensity score, we compute location-specific weights for untreated locations that are given by $w_l(Z_{lt^*-1}) = \hat{p}(Z_{lt^*-1})/(1-\hat{p}(Z_{lt^*-1}))$ where $\hat{p}(z)$ denotes the estimated propensity score; we also normalize the weights so that their average is equal to one (that is, the final weights we use are $w_l(Z_{lt^*-1})/\bar{w}$ where $\bar{w}$ is the sample average of the weights among untreated locations).  The second step is a regression adjustment step.  Given the weights that from the first step, we run a weighted regression of particular outcomes of interest on the same set of covariates as above using the set of untreated locations.  By using the first step weights, this step puts more weight on untreated locations that have pre-treatment characteristics that are relatively more common among treated locations.  The final step is to impute untreated potential outcomes for the treated group.  Given the estimated parameters from the previous step, untreated potential outcomes for the treated group can be imputed by calculating predicted values from the previous regression using the pre-treatment characteristics of treated locations.  We estimate $ATT$s by calculating the difference between average observed outcomes and average imputed untreated potential outcomes for the treated group.  One important advantage of this approach is that it is doubly robust in the sense that our estimates of $ATT$s are consistent if \textit{either} the propensity score model \textit{or} the outcome regression model is correctly specified.  A main reason that this is attractive in this context is that it sidesteps needing to estimate a full pandemic model in order to estimate effects of policies, particularly in the case where the propensity score model is correctly specified.

For the results below, we report daily estimates of treatment effects for various outcomes (e.g., tests, confirmed cases, and bounds on actual cases) from March 18 to May 9.  We set the policy implementation date to be April 1 (i.e., in the notation of the paper, we are setting $t^*$ to be April 1).  For dates after April 1, we condition on county-specific covariates from March 25 (7 days before).  Then, we report daily estimated ATTs through May 9, one week after Tennessee's open-testing policy ended and as Tennessee was adjusting to a more targeted testing strategy.  We also report ``pre-treatment'' estimates going back to March 18.  For dates before April 1, we condition on covariates seven days before that particular date; for example, for estimates on March 28, we condition on covariates from March 21.  We report results both for cumulative tests and confirmed cases and the seven day moving average of tests and confirmed cases.\footnote{Our main interest is in understanding the effect of the policy on the cumulative outcomes.  However, the results using the seven day moving average are more comparable to each other across dates and make it easier to see the timing of policy effects.}   One could make different choices for which date to set as the policy implementation date besides April 1, but there are tradeoffs here.  Using later dates makes it harder to find comparison counties with the same pandemic-related characteristics as treated counties because Tennessee's large increase in testing makes it harder to find comparison counties that had conducted as many tests.  Along these lines, later dates would tend to understate the full effects of Tennessee's expanded testing as well.  However, setting the date too early can involve inappropriately using comparison units just because the pandemic had not fully started yet (for example, in mid-March, many counties had not had any confirmed Covid-19 cases yet).  Ideally, we would like to use as early of a date as possible such that the pandemic has actually started in all counties.  Using April 1, to some extent, balances these tradeoffs.  In the Supplementary Appendix, we alternatively use March 25 and April 18 (the date when open-testing was implemented) and find broadly similar results.

In our application, the variation in the policy is at the state level, and it is common practice in empirical work to cluster at the level of the treatment.  However, this approach is made difficult in our setting because we have only a single treated state.  Instead, our approach is to cluster at the county level.  The disadvantage of this approach is that it is not robust to state-level ``common shocks.''  A main source of these common shocks would be other policies implemented by states to limit the spread of Covid-19 (e.g., stay-at-home orders or school closures), and, below, we carefully check the mix and timing of other state-level policies for Tennessee and its surrounding states.  Alternatively, our approach can be viewed as a conditional inference procedure where we condition on the state-level common shocks.  See Section 5.1 of \citet{roth-santanna-bilinski-poe-2022} for a discussion along these lines in a related context.  Finally, in the Supplementary Appendix, we also provide results using alternative strategies that include (i) a matching estimator along the lines of \citet{ho-imai-king-stuart-2007,abadie-spiess-2021}, (ii) varying the start date of the policy, and (iii) including the change in confirmed cases over time in pre-treatment periods as an additional covariate (including the change in covariates can be motivated by a pandemic model that includes location-specific unobserved heterogeneity in transmission rates).

\subsection{Challenges to Identification}

As is clear from \Cref{fig:state-tests-cases}, there are important differences between Tennessee and its surrounding states in terms of population and pandemic related characteristics.  These differences show up before Tennessee's expanded testing began which implies that that they are not due to the policy itself.  \Cref{tab:balance} reports summary statistics and covariate balance measures for the county-level data that we use in our main results.  The table provides information both for the raw data and after re-weighting using the weights coming from the propensity score (as described above) and on March 25 (which is seven days before the date we set as the policy implementation date).  In the raw data, on March 25, on average untreated counties had run about 10\% more tests than counties in Tennessee and had confirmed about 80\% more cases than counties in Tennessee.  There is a much bigger difference in terms of the average number of deaths in counties in Tennessee relative to untreated counties.\footnote{This large difference is mainly driven by high death rates in some counties in surrounding states early in the pandemic.  For example, the county with the highest death rate in our sample as of March 25 was Baker County, GA, where there had been 1.2 deaths per 1000 people (which is triple the rate of the second highest county in our sample).  Baker County is a rural county bordering Albany, GA which was one of the most notable Covid-19 hotspots in the entire country during the period that we consider (\citet{schrade-edwards-2020}).  On March 25, the highest death rates tend to be in rural counties (e.g., among the top 10 highest death rates, no county has a population over 25,000) and also tend to be geographically concentrated in Georgia (12 out of the top 20 county-level death rates) and Mississippi (5 of the top 20 death rates).}  Re-weighting untreated counties makes a notable difference; after re-weighting, the means of all variables are essentially identical.

Next, we briefly discuss the timing of other policy decisions made by Tennessee and the six comparison states.  We list the timing of implementing major policies across states in \Cref{tab:timing}.  The timing of other policies is important in this context because (i) states implemented a number of policies in response to the Covid-19 pandemic and (ii) if the policies themselves or the timing of implementing these policies differed substantially across Tennessee and its surrounding states, then our results would mix together the effects of Tennessee's testing policy with other policy differences between Tennessee and the six comparison states.

The timing of main policies across the states that we consider is, in general, very similar.  In particular, besides testing, the timing of Tennessee's policies were virtually identical to the timing in Alabama, Georgia, and Mississippi.  There are modest differences with Kentucky and North Carolina (though these differences are most notable for business closure policies and gathering restrictions which are the two least well-defined policies that we consider).  The most notable policy difference is that Arkansas did not implement a stay-at-home order; we show in the Supplementary Appendix that the results are not sensitive to excluding Arkansas from the comparison group.  These close similarities in terms of other policies across states provides one piece of evidence in favor of interpreting our results below as being due to Tennessee's expanded testing.

Finally, it is important to be clear that, although our identification arguments used terminology from the treatment effects literature such as ``treated and untreated potential outcomes,'' the effects that we estimate in this section are not relative to a ``no testing at all'' counterfactual.  Instead, we are estimating treatment effects of the policy that Tennessee implemented relative to a counterfactual policy where Tennessee's testing expanded similarly to its surrounding states (and implemented the same mix of other policies such as school closures and stay-at-home orders).

\subsection{Results} \label{sec:results}

\subsubsection*{Descriptive Bounds on the Per Capita Number of Actual Cases across Counties}

To start with, we compute bounds on the per capita number of actual Covid-19 cases across counties in Tennessee, and, for comparison, for counties in Alabama (which was arguably experiencing the most similar path of Covid-19 cases, tests, and policies among Tennessee's surrounding states).  These results are available in \Cref{fig:infection-bounds} in the Supplementary Appendix for March 31 and April 25.  The bounds are informative but still fairly wide.  To give an example, the lower bound for Davidson County (the county where Nashville is located) is that 0.41\% of county residents had had Covid-19 by April 25; the upper bound indicates that 9.8\% had had Covid-19 by the same date.  More generally, for both states, the bounds are somewhat narrower by April 25 than they were on March 31 -- this should not be surprising as the number of tests had increased substantially in both states over time.  In general, the bounds tend to be narrower for counties in Tennessee than for counties in Alabama.  We provide some additional discussion of these results as well as results across states in the Supplementary Appendix.

\subsubsection*{Main Results: Policy Effects of Tennessee's Expanded Testing}
  
This section considers the effect of Tennessee's expanded testing on observed outcomes including tests and confirmed cases as well bounds on actual cases.  First, we consider the effect of the policy on the number of tests and on the number of confirmed cases.  These results are available in \Cref{fig:tests-cases-es}.  We estimate that Tennessee's expanded testing policy increased the cumulative number of Covid-19 tests run in the state by May 9 by about 22 tests per 1000 people relative to counties with similar populations, tests, confirmed cases, and Covid-19 deaths in surrounding states prior to April 1.  This is a large increase; in particular, our estimate indicates the policy slightly more than doubled the number of tests relative to a counterfactual where testing in Tennessee followed a similar path as in surrounding states (we calculate this by dividing our estimate of the $ATT$ by the average untreated potential outcome for Tennessee which is available as a byproduct of our estimation strategy).  Moreover, in panel (b) of \Cref{fig:tests-cases-es}, it is clear that the main expansion of testing in Tennessee corresponds to its open-testing policy which began on April 18.

Next, we move to the effects of expanded testing on confirmed cases.  We estimate that Tennessee's policy decreased the cumulative number of confirmed cases by about 0.9 per 1000 people in Tennessee relative to what they would have been in the absence of the policy.  Once again, this is a large effect; this estimate is a decrease in cumulative confirmed cases by about 34\%.  The timing of the decrease in confirmed cases is also in line with the policy decreasing the number of actual Covid-19 cases.  In particular, it appears that there is some lag in testing expanding and lower confirmed cases.  For example, testing was expanding in Tennessee by the end of March, but confirmed cases only start to go down about a week into April and the largest decreases are later.  Similarly, the timing of the largest decreases in confirmed cases (during the last half of April and early May) roughly corresponds to the timing of Tennessee's major expansion of testing due to its open-testing policy.

We provide our estimated bounds on the effect of Tennessee's policy on actual Covid-19 cases in \Cref{fig:bounds}.\footnote{The top panel of the figure, that contains estimates for the upper bound, also includes a one-sided 90\% confidence interval for the upper bound.  This is essentially immediately available from our inference procedure for the number of confirmed cases due to the upper bound being a linear functional of the effect of the policy on confirmed cases.  The lower bound is a more complicated functional (see the expression for the lower bound in \Cref{prop:2}).  Given that (i) our primary interest is on inference regarding the upper bound, and (ii) the estimated bound is extremely negative so that sampling variance is likely to be small in magnitude relative to the value of the bound itself, we do not report a confidence interval for the lower bound.}  The lower bound amounts to a huge reduction in actual Covid-19 cases due to the policy; however, as discussed above, this lower bound occurs under the conditions that no untested individuals in Tennessee have had Covid-19, but that, in the absence of the policy, the fraction of untested individuals that have Covid-19 is the same as the fraction of tested individuals that have Covid-19.  This particular scenario seems unlikely and, as discussed in \Cref{rem:shorter-bounds}, the lower bound could be tightened under additional assumptions.  However, our primary interest is in the upper bound, and this is what we primarily focus on below.  The results in \Cref{fig:bounds} suggest that, after about April 6, Tennessee's expanded testing was reducing the number of Covid-19 cases in Tennessee.   From \Cref{prop:3}, in combination with the expression in \Cref{eqn:CcondT}, the upper bound is a scaled version of the difference in confirmed cases across counties in Tennessee relative to counties in the surrounding states with similar pre-policy characteristics.  Therefore, the top panel of \Cref{fig:bounds} is very similar to Panel (c) of \Cref{fig:tests-cases-es}.  We estimate that by May 9, Tennessee's expanded testing decreased actual cases by at least 1.18 per 1000 people (the corresponding one-sided 90\% confidence interval is 0.85 cases per 1000).  Even at the upper bound, this indicates a relatively large effect of Tennessee's testing policy.  These results suggest that Tennessee's policy led to fewer actual Covid-19 cases.  To get a sense of the magnitude of these effects, it is natural to compare this estimate to the (observed) cumulative number of confirmed cases in Tennessee on May 9 which was 1.75 per 1000.\footnote{For the previous outcomes such as tests and confirmed cases above, we additionally reported the the percentage change in the outcome due to the policy.  This is more challenging here because our bounds on the effect of the policy do not pin down what the average number of actual Covid-19 cases would have been in the absence of the policy.  In the Supplementary Appendix, we carry out an additional exercise where we back out the number of actual cases in the absence of the policy using a combination of (i) additional (potentially much) stronger assumptions and (ii) estimates of the overall infection rate in Tennessee across time from the IHME.  In this case, we estimate  that the policy reduced the actual number of Covid-19 cases in Tennessee by about 15\%; and although this involves a number of extra assumptions, it seems like a useful baseline for interpreting the magnitudes of the estimates here.}  At a minimum the results here suggest that (i) there is strong evidence that expanded testing did decrease the actual number of Covid-19 cases in Tennessee, and (ii) even the upper bound on the effect of the policy should be interpreted as a non-trivial reduction in Covid-19 cases.  At a higher level, that the upper bound is negative is driven by the fact that confirmed cases appear to have decreased in Tennessee due to the open-testing policy.  This decrease in confirmed cases in the presence of an increase in total tests is a strong piece of evidence that Tennessee's expanded testing did decrease the number of actual cases -- even if we are not able to provide plausible assumptions that lead to point identification.

\subsubsection*{Discussion}

The results in this section indicate that, even under weak assumptions on the number of Covid-19 cases among untested individuals, there is considerable evidence that Tennessee's expanded testing reduced actual Covid-19 cases relative to the number of cases that would have occurred without the expanded testing.  These results come with standard caveats in the context of policy evaluation that are worth mentioning here.  First, these results are local both in geography and in time.  For example, our results would not necessarily generalize to later in the pandemic when testing was more widely available.  It is also not immediately evident how alternative policies such as a much larger increase in testing would have affected the spread of Covid-19 over the first few months of the pandemic.  Second, our results cannot pin down the mechanism through which Tennessee's policy affected Covid-19 cases. It is not clear whether the effects of the policy are fully driven by the expansion of testing \textit{per se}.  Thus, it is not clear if, holding the number of tests fixed, allowing individuals to self-select into testing (as in Tennessee's policy) is more effective than other policies where testing would be targeted to exposed individuals or high-risk groups.%

Besides our main results in the paper, we also provide some additional evidence on the effectiveness of expanded testing in the Supplementary Appendix.  We show that Tennessee's expanded testing appears to have reduced the number of Covid-19 deaths --- this effect is point identified in our context and, in turn, suggests that the policy reduced the number of cases.  Similarly, although we do not observe county-level hospitalizations, at the state-level, Covid-19 hospitalizations appear to be reduced in Tennessee relative to other states being considered.  Finally, we provide additional estimates of the effect of the policy on trips to work in \Cref{fig:trips} in the Supplementary Appendix.  These results indicate that the number of trips to work was higher under the policy in the first week of May than it would have been in the absence of the policy.  This is arguably in line with our previous results and can be rationalized in a model where individuals increase their travel when their risk of becoming infected with Covid-19 is lower.

To conclude, we briefly consider the magnitude and policy relevance of our estimated effects.  By May 9, 241 people had died from Covid-19 in Tennessee.  At the upper bound on the effect of the policy (our most conservative estimate), we estimate that there were 8058 fewer actual cases in Tennessee than there would have been without the policy.  The Institute for Health Metrics and Evaluation (IHME) provides an estimate of the infection fatality rate in Tennessee on April 15 of 0.00988 (i.e., slightly less than 1\%).  Given this IFR, we would estimate that Tennessee's policy reduced the number of Covid-19 deaths by at least 80 in Tennessee (indicating about 25\% fewer deaths due to the policy).  In terms of a cost-benefit analysis, Tennessee expected to pay about \$100 per test (\citet{farmer-2020}).  Aggregated up to the whole state, we estimate that Tennessee ran close to 136,000 additional tests under the policy for a total cost of about \$14 million dollars.   This suggests an extremely high return on investment to expanding testing for Tennessee during the early part of the pandemic.  In fact, the costs of expanded testing were so low that essentially any reduction in the number of deaths due to the policy would justify its cost.

\FloatBarrier

\section{Conclusion} \label{sec:conclusion}

In this paper, we have proposed a new approach to evaluating the effects of Covid-19 related policies.  Our approach is particularly useful for evaluating policies early in the pandemic when testing was not widely available.  Our idea is to combine standard policy evaluation identifying assumptions with relatively mild assumptions to deal with actual Covid-19 cases in a particular location not being observed.  This strategy leads to bounds on the effect of the policy on actual Covid-19 cases and point identification of the effect of the policy on other observed outcomes.  The bounds on actual Covid-19 cases can be informative especially for policies that reduce the number of confirmed Covid-19 cases while not decreasing the number of Covid-19 tests.

We used this approach to study the effect of Tennessee's expanded testing early in the pandemic that made testing available to anyone who wanted a test.  %
Overall, our results indicate that, even under relatively weak assumptions, Tennessee's policy appears to have decreased the number of total and confirmed cases in Tennessee.  In this sense, it seems that Tennessee's policy had the intended effects.

{ \singlespacing 
\setstretch{1}
\setlength\bibitemsep{0pt}
\printbibliography}

\appendix

\section{Proofs} \label{sec:proofs}

\begin{proof}[\textbf{Proof of \Cref{prop:1}}]
  The result follows because
  \begin{align*}
    ATT_Y(Z_{lt^*-1}) &= \E[ Y_{lt^*}(1) | Z_{lt^*-1}, D_l=1] - \E[Y_{lt^*}(0) | Z_{lt^*-1}, D_l=1] \\
                      &= \E[ Y_{lt^*}(1) | Z_{lt^*-1}, D_l=1] - \E\Big[ \E[Y_{lt^*}(0) | Z^*_{lt^*-1}, D_{l} = 1 ] \big| Z_{lt^*-1}, D_l=1 \Big] \\
                      &= \E[ Y_{lt^*}(1) | Z_{lt^*-1}, D_l=1] - \E\Big[ \E[Y_{lt^*}(0) | Z^*_{lt^*-1}, D_{l} = 0 ] \big| Z_{lt^*-1}, D_l=1 \Big] \\
                      &= \E[ Y_{lt^*}(1) | Z_{lt^*-1}, D_l=1] - \E\Big[ \E[Y_{lt^*}(0) | Z^*_{lt^*-1}, D_{l} = 0 ] \big| Z_{lt^*-1}, D_l=0 \Big] \\
                      &= \E[ Y_{lt^*}(1) | Z_{lt^*-1}, D_l=1] -  \E[Y_{lt^*}(0) | Z_{lt^*-1}, D_l=0] \\
                      &= \E[ Y_{lt^*} | Z_{lt^*-1}, D_l=1] -  \E[Y_{lt^*} | Z_{lt^*-1}, D_l=0]
  \end{align*}
  which is the result.  The first equality is the definition of $ATT_Y(Z_{lt^*-1})$; the second equality holds by the law of iterated expectations (the outer expectation averages over the distribution of $C_{lt^*-1}(0)$ conditional on $Z_{lt^*-1}$ and $D_l=1$); the third equality holds by \Cref{ass:unc}; the fourth holds by \Cref{ass:mar}; the fifth equality holds by the law of iterated expectations; and the sixth equality holds because $Y_{lt^*}(1)$ is the observed outcome when $D_{l}=1$ and $Y_{lt^*}(0)$ is the observed outcome when $D_l=0$.  The result for $ATT_Y$ holds immediately by averaging over $ATT_Y(Z_{lt^*-1})$ over the distribution of $Z_{lt^*-1}$ conditional on $D_l=1$.
\end{proof}

\begin{proof}[\textbf{Proof of \Cref{prop:attc-decomp}}]
  First, recall that
  \begin{align*}
    ATT_C(Z_{lt^*-1}) &=  \E[C_{lt^*}(1) - C_{lt^*}(0) | Z_{lt^*-1}, D_l=1] \\
                    &= \E[C_{lt^*} | Z_{lt^*-1}, D_l=1]  - \E[C_{lt^*} | Z_{lt^*-1}, D_l=0]  \\
                    &= \P(C_{ilt^*}=1|Z_{lt^*-1}, D_l=1) - \P(C_{ilt^*}=1|Z_{lt^*-1},D_l=0) 
  \end{align*}
  where the second equality by using the same arguments as in the proof of \Cref{prop:1} and the third equality holds by the definition of $C_{lt^*}$.
  Omitting the dependence on $Z_{lt^*-1}$ for notational simplicity, and then plugging in from \Cref{eqn:covid-bound-1,eqn:CcondT} and the definitions of $\gamma_d(Z_{lt^*-1})$, $\phi_d(Z_{lt^*-1})$, and $\tau_d(Z_{lt^*-1})$ further implies that
  \begin{align} \label{eqn:covid-policy-bound}
    \P(C_{ilt^*}(1)=1|D_l=1) - \P(C_{ilt^*}(0)=1|D_l=1) &= \gamma_1 - \gamma_0 + \phi_1 (1-\tau_1) - \phi_0 (1-\tau_0)
  \end{align}
  which is the result.
  \end{proof}
  
  \begin{proof}[\textbf{Proof of \Cref{prop:2}}]
  Starting from \Cref{eqn:covid-policy-bound}, $\gamma_1$, $\gamma_0$, $\tau_1$, and $\tau_0$ are point identified but $\phi_1$ and $\phi_0$ are not (again we are omitting conditioning on $Z_{lt^*-1}$ to minimize notation).  Bounds on the effect of the policy on actual Covid-19 cases arise from restrictions on these terms.  In particular, \Cref{ass:covid-bound} says that, for $d \in \{0,1\}$,
  \begin{align*}
    0 \leq \phi_d \leq \P(C_{ilt^*}=1|T_{ilt^*}=1,D_l=d)
  \end{align*}
  $C^{B,U}_{lt^*}(Z_{lt^*-1})$, the upper bound in the proposition, comes from setting $\phi_1 = \P(C_{ilt^*}=1|T_{ilt^*}=1,D_l=1)$ (its maximum value under \Cref{ass:covid-bound}) and from setting $\phi_0=0$.  $C^{B,L}_{lt^*}(Z_{lt^*-1})$, the lower bound in the proposition, comes from setting $\phi_1 = 0$ and from setting $\phi_0 = \P(C_{ilt^*}=1|T_{ilt^*}=1,D_l=0)$ (its maximum value under \Cref{ass:covid-bound}).  In addition, the expressions provided in the proposition require noting that $\P(C_{ilt^*}=1|T_{ilt^*}=1,D_l=d) = \gamma_d/\tau_d$.
  The bounds on $ATT_C$ arise from averaging over the bounds for $ATT_C(Z_{lt^*-1})$ as discussed in the text.
\end{proof}

Next, we provide an auxiliary result that is useful for proving \Cref{prop:3}.

\begin{lemma} \singlespacing \label{lem:covid-policy-bound-2} Under \Cref{ass:covid-bound,ass:mar,ass:covid-unc,ass:pol-bound},
  \begin{align}
    \P(C_{ilt^*}=1, T_{ilt^*}=0 | Z_{lt^*-1},D_l=1) \leq \P(C_{ilt^*}=1,T_{ilt^*}=0|Z_{lt^*-1},D_l=0) 
  \end{align}
\end{lemma}
\begin{proof}
  To show the result (and omitting conditioning on $Z_{lt^*-1}$), notice that
  \begin{align*}
    \P(C_{ilt^*}=1,T_{ilt^*}=0|D_l=1) &= \P(C_{ilt^*}(1)=1,T_{ilt^*}(1)=0|D_l=1) \nonumber \\
                                    & \leq \P(C_{ilt^*}(0) = 1, T_{ilt^*}(0) = 0 | D_l=1) \nonumber\\
                                    &= \P(C_{ilt^*}=1,T_{ilt^*}=0|D_l=0) 
  \end{align*}
  where the first equality holds because treated potential outcomes are observed outcomes when $D_l=1$, the second line holds by \Cref{ass:pol-bound}, and third line holds by \Cref{ass:covid-unc,ass:mar}.
\end{proof}

\begin{proof}[\textbf{Proof of \Cref{prop:3}}]
    Continuing to omit conditioning on covariates to simplify the notation, first, notice that, for $d \in \{0,1\}$, $\phi_d(1-\tau_d) = \P(C_{ilt^*} = 1, T_{ilt^*}=0|D_l=d)$.  Thus, \Cref{lem:covid-policy-bound-2} implies that
    \begin{align} \label{eqn:phi1-bound}
        \phi_1(1-\tau_1) \leq \phi_0 (1-\tau_0)
    \end{align}

  Next, following the same logic as in the proof of \Cref{prop:2} (see \Cref{eqn:covid-policy-bound} in particular), the lower bound arises by making $\phi_1$ as small as possible while making $\phi_0$ as large as possible.  \Cref{ass:pol-bound} does not provide any identifying power for the lower bound though (see \Cref{eqn:phi1-bound}) so the lower bound remains unchanged.
    
    Under \Cref{ass:pol-bound}, from \Cref{eqn:phi1-bound}, it follows that $\phi_1(1-\tau_1) - \phi_0 (1-\tau_0) \leq 0$, and, plugging this into   \Cref{eqn:covid-policy-bound}, implies that
  \begin{align*}
    & \P(C_{ilt^*}(1) = 1| D_l=1) - \P(C_{ilt^*}(0) = 1 | D_l=1) \leq \gamma_1 - \gamma_0 
  \end{align*}
  which implies the result for the upper bound of $ATT_C(Z_{lt^*-1})$.  The result for $ATT_C$ holds by averaging over the $Z_{lt^*-1}$ in $ATT_C(Z_{lt^*-1})$.  
\end{proof}

\pagebreak

\section{Figures and Tables}

\begin{figure}[h]
    \centering
    \caption{Tests and Confirmed Cases by State}
    \label{fig:state-tests-cases}
    \begin{subfigure}{.49\textwidth}
        \includegraphics[width=\textwidth]{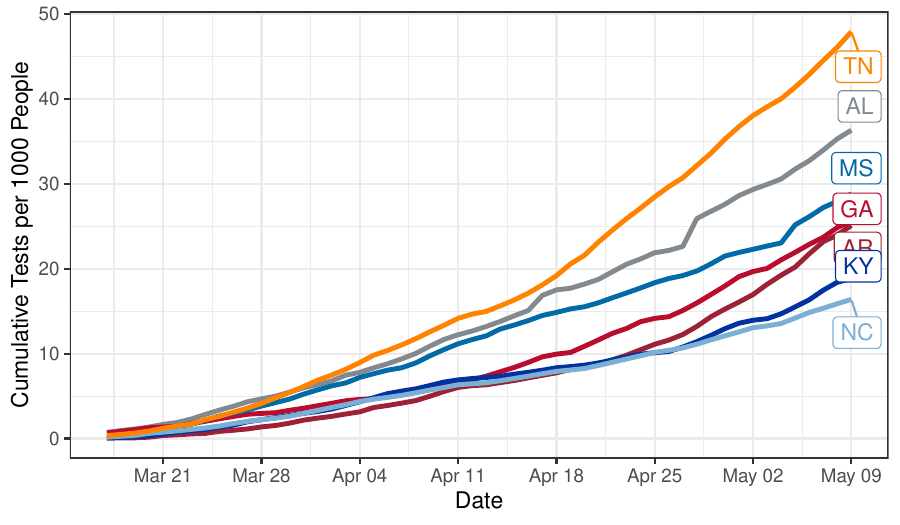}  
        \caption{Cumulative Tests}
    \end{subfigure}
    \begin{subfigure}{.49\textwidth}
        \includegraphics[width=\textwidth]{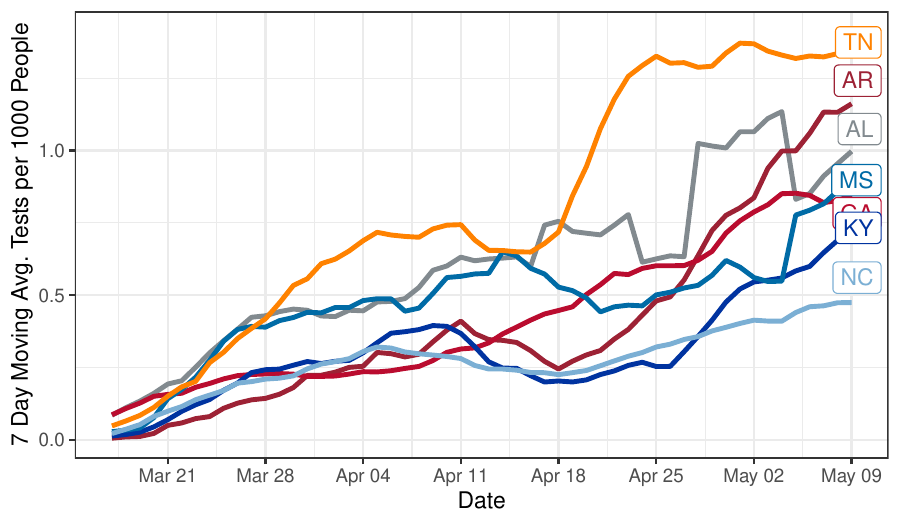}
        \caption{7 Day Moving Avg. Tests}
    \end{subfigure}
    \begin{subfigure}{.49\textwidth}
        \includegraphics[width=\textwidth]{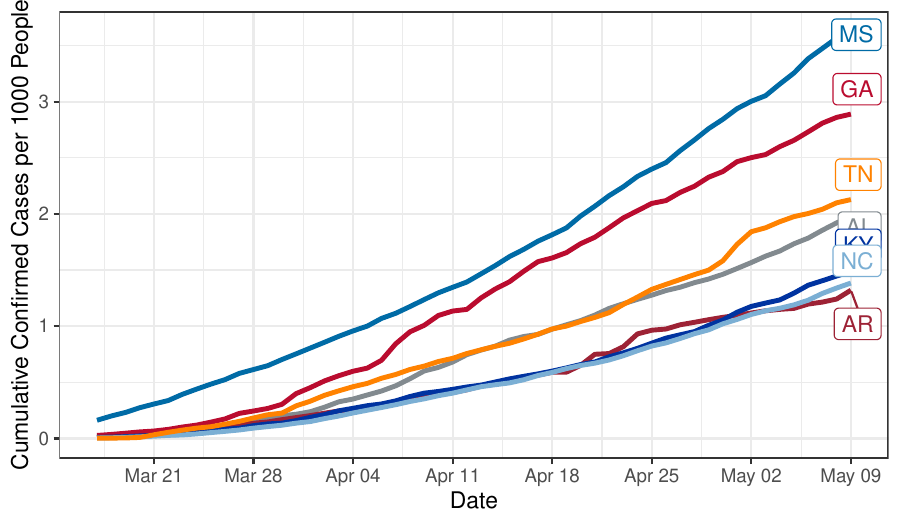}  
        \caption{Cumulative Confirmed Cases}
    \end{subfigure}
    \begin{subfigure}{.49\textwidth}
        \includegraphics[width=\textwidth]{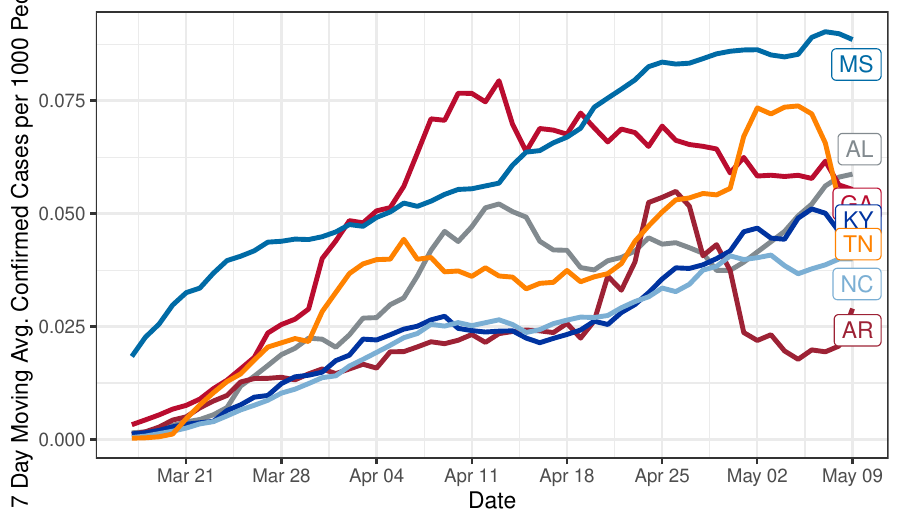}
        \caption{7 Day Moving Avg. Confirmed Cases}
    \end{subfigure}
    \begin{justify}
        { \footnotesize \textit{Notes:} The figure provides cumulative and seven day moving averages of tests and confirmed cases for Tennessee, Alabama, Arkansas, Georgia, Kentucky, Mississippi, and North Carolina from March 18 to May 9, 2020.
        
        \noindent \textit{Sources:} CDC COVID Data Tracker, \url{https://covid.cdc.gov/covid-data-tracker/#trends_newtestresultsreported}}
    \end{justify}
\end{figure}

\pagebreak

\renewcommand{\arraystretch}{1.3}
\begin{table}[t]
\centering
\caption{Covariate Balance}
\label{tab:balance}
\begin{tabular}{@{}rcccc@{}}
\toprule
                                                & Mean Treated         & Mean Untreated       & Std. Mean Diff       & Var Ratio            \\ \midrule
\multicolumn{1}{l}{{\ul \textbf{Raw Data}}}     & \multicolumn{1}{l}{} & \multicolumn{1}{l}{} & \multicolumn{1}{l}{} & \multicolumn{1}{l}{} \\
Tests                                           & 1.802                & 1.997                & -0.158               & 0.404                \\
Confirmed Cases                                 & 0.120                & 0.218                & -0.580               & 0.227                \\
Deaths                                          & 0.0002               & 0.0075               & -6.557               & 0.0003               \\
Log Population                                  & 10.52                & 10.28                & 0.234                & 0.907                \\ \midrule
\multicolumn{1}{l}{{\ul \textbf{Re-weighted Data}}} &                      &                      &                      &                      \\
Tests                                           & 1.802                & 1.799                & 0.002                & 0.582                \\
Confirmed Cases                                 & 0.120                & 0.120                & 0.002                & 0.800                \\
Deaths                                          & 0.0002               & 0.0002               & 0.006                & 0.550                \\
Log Population                                  & 10.52                & 10.53                & -0.006               & 0.805                \\ \bottomrule
\end{tabular}

\begin{justify}
{ \footnotesize \textit{Notes:} The table provides summary statistics and covariate balance measures for the underlying data and for the re-weighted data.  Besides county-level population, all reported statistics are cumulative values for each variable per 1000 people up to March 25 (which is seven days before April 1 --- the date we use as the implementation date of the policy).  The column labeled ``Std. Mean Diff'' reports the standardized difference in means between the treated and comparison groups (i.e., the difference between the means for each group divided by the standard deviation of the same variable for the treated group), and the column labeled ``Var Ratio'' reports the variance of that variable for the treated group divided by the variance for the untreated group. } 
\end{justify}
\end{table}

\pagebreak

\begin{table}[t]
\centering
\caption{Timing of Other Policies}
\label{tab:timing}
\small
\begin{tabular}{@{}rccccc}
\toprule
 & Emergency Dec. & Schools Closed & Stay-at-Home & Business Closure & Gathering Rest. \\ \midrule
TN Start & March 12 & March 16 & April 2 & April 1 & March 23 \\
TN End & - & - & May 1 & May 30 & May 1 \\
AL Start & March 13 & March 19 & April 4 & $^*$ & March 19$^c$ \\
AL End & - & - & April 30 & - & May 15 \\
AR Start & March 11 & March 16 & none & $^*$ & March 27 \\
AR End & - & - & - &  & May 11 \\
GA Start & March 14 & March 18 & April 3 & $^*$ & March 24 \\
GA End & - & - & May 1 & - & - \\
KY Start & March 6 & March 16 & March 26$^a$ & April 26 & March 19 \\
KY End & - & - & - & May 11$^b$ & - \\
MS Start & March 14 & March 19 & April 3 & April 3 & March 24 \\
MS End & - & - & April 27 & - & - \\
NC Start & March 10 & March 16 & March 30 & March 30 & March 12$^c$ \\
NC End & - & - & May 8$^b$ & - & - \\ \bottomrule
\end{tabular}
\begin{justify}
    {\footnotesize \textit{Notes:} The table reports the timing that various policies were implemented in the states considered in the paper.  The data comes from \citet{fullman-et-al-2021}, and we use their classification scheme.  In some cases, multiple versions of the same policy are reported for the same state over the period that we consider.  In these cases, we default to reporting the first state-wide mandated policy except where noted in the table.  The column labeled ``Emergency Dec.'' provides the date that the state declared an emergency; the column labeled ``Schools Closed'' records the date when schools were closed state-wide; the column labeled ``Stay-at-Home'' provides the date when a mandatory stay-at-home order was implemented; the column labeled ``Business Closure'' refers to non-essential business closures though other business restrictions were imposed in some states; and the column labeled ``Gathering Rest.'' provides the date when the state imposed some state-wide gathering restriction.  The additional notation in the table has the following meaning: ``$-$'' -- policy did not end before May 31, 2020, $^*$ -- several business closure policies but none were  classified as a non-essential business closure policy, $^a$ -- policy recommendation rather than mandate, $^b$ -- policy eased rather than removed, $^c$ -- multiple versions of this policy were enacted and the earliest date is reported in the table.}
\end{justify}
\end{table}

\pagebreak

\begin{figure}[t]
  \begin{center}
    \caption{Policy Effects on Tests and Confirmed Cases}
    \label{fig:tests-cases-es}
    \begin{subfigure}{.49\textwidth}
        \includegraphics[width=\textwidth]{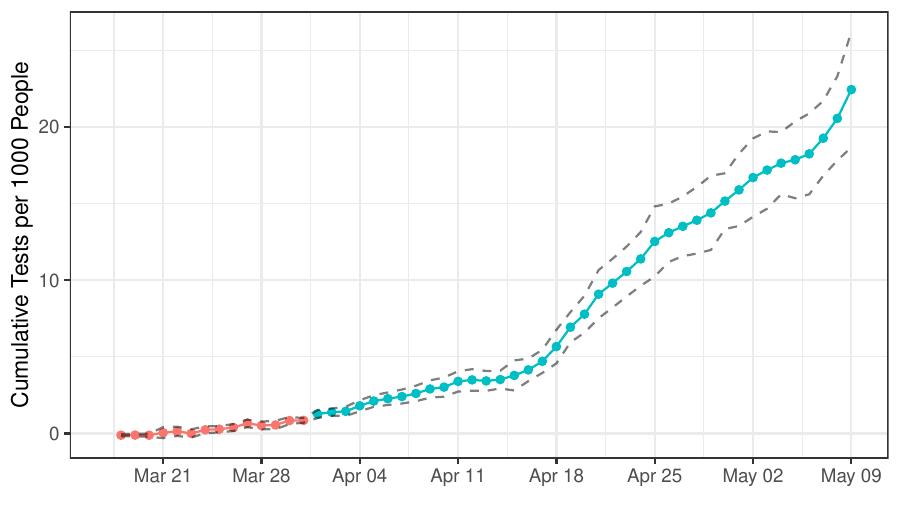}
        \caption{Cumulative Tests}
    \end{subfigure}
    \begin{subfigure}{.49\textwidth}
        \includegraphics[width=\textwidth]{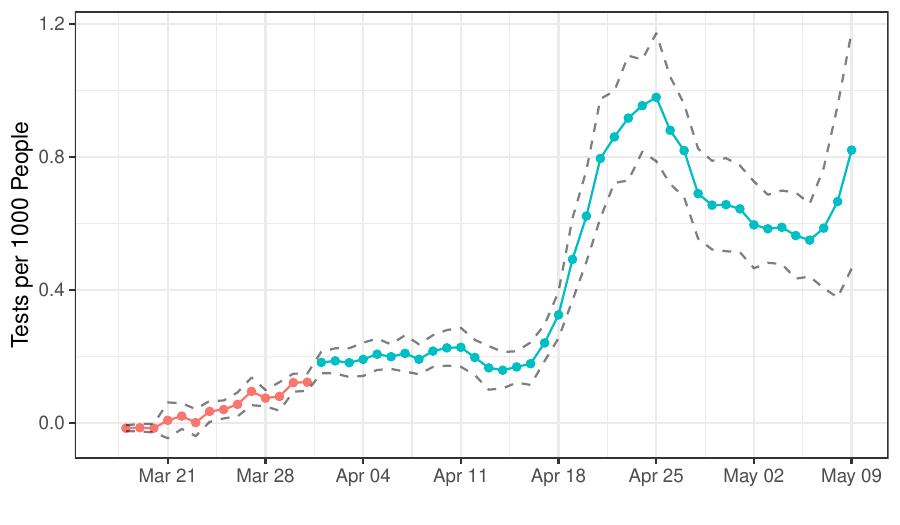}
        \caption{7 Day Moving Avg. Tests}
    \end{subfigure}
    \begin{subfigure}{.49\textwidth}
        \includegraphics[width=\textwidth]{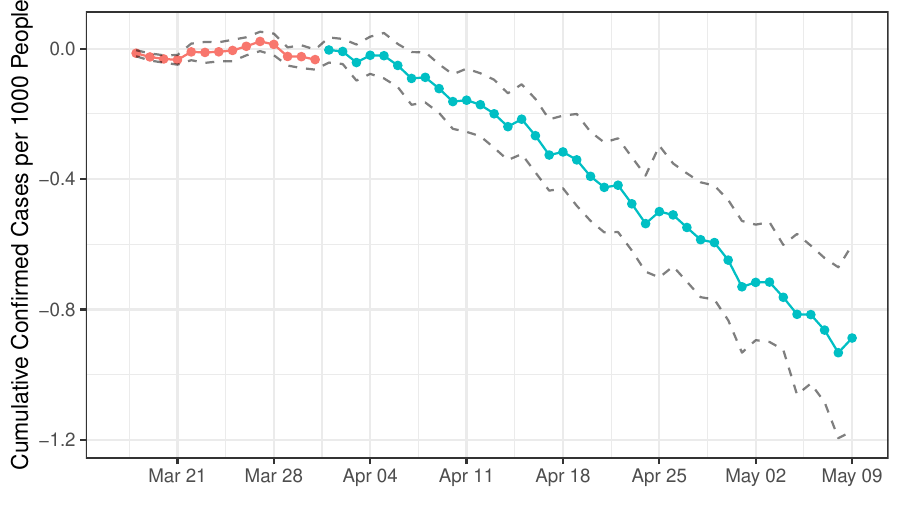}
        \caption{Cumulative Confirmed Cases}
    \end{subfigure}
    \begin{subfigure}{.49\textwidth}
        \includegraphics[width=\textwidth]{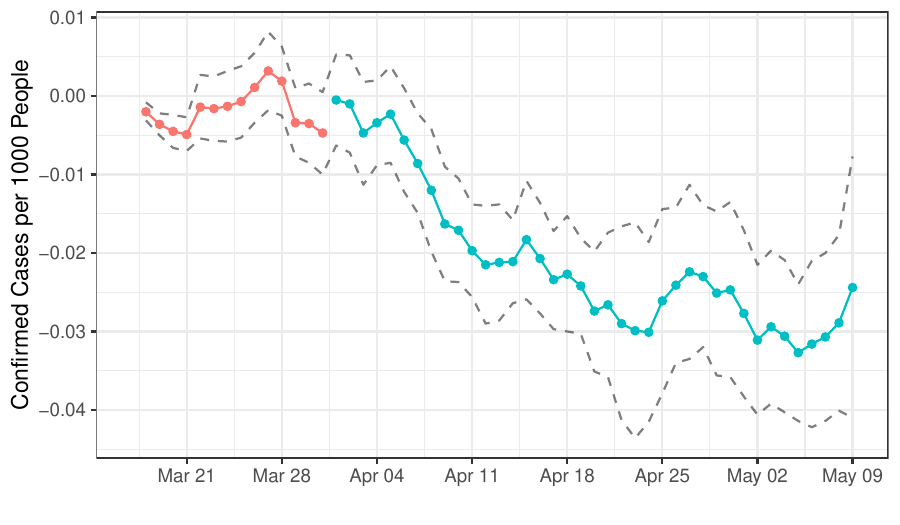}
        \caption{7 Day Moving Avg. Confirmed Cases}
    \end{subfigure}
  \end{center}
  \begin{justify}
        { \footnotesize \textit{Notes:} The figure provides estimates of the effects of Tennessee's expanded testing policy on cumulative tests per 1000 people, the seven day moving average of tests per 1000 people, cumulative confirmed cases per 1000 people, and the seven day moving average of confirmed cases per 1000 people using the approach described in the text.  The red points in the figure are estimates before April 1 while the blue points are for after April 1.  The dashed line provides a 90\% confidence interval.}
    \end{justify}
\end{figure}

\pagebreak

\begin{figure}[t]
    \centering
    \caption{Bounds on Policy Effects on Actual Covid-19 Cases}
    \label{fig:bounds}
    \includegraphics[width=.75\textwidth]{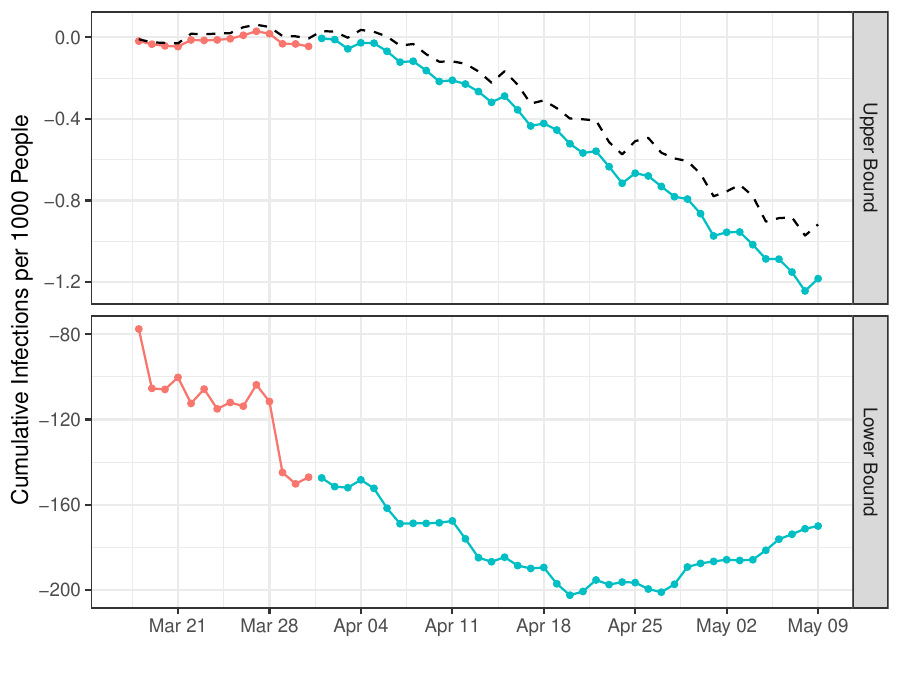}
    \begin{justify}
        { \footnotesize \textit{Notes:} The figure provides estimates of the bounds on the number of actual Covid-19 cases per 1000 people using the approach developed in the paper.  The upper bound is provided in the top panel and the lower bound is provided in the bottom panel.  Notice that the values on the y-axis change between the two panels.  The top panel also contains a one-sided 90\% confidence interval for the upper bound on the effect of the policy.}
    \end{justify}
\end{figure}

\end{document}